\documentclass[a4paper,11pt]{amsart}
\usepackage{fullpage}
\usepackage{stylefile}
\usepackage{setspace}
\title{Checking the strict positivity of Kraus maps is NP-hard}
\author{St\'ephane Gaubert}
\address{INRIA and CMAP, \'Ecole Polytechnique. 91128 Palaiseau Cedex}
\email{Stephane.Gaubert@inria.fr}
\author{Zheng Qu}
\address{{School of Mathematics, University of Edinburgh, Edinburgh, EH9 3JZ, UK}}
\email{{zheng.qu@ed.ac.uk}}
\thanks{This work was partially supported by the PGMO (Gaspard Monge) Program of FMJH (Fondation Math\'ematique Jacques Hadamard) and EDF. It was carried out when the second author was with CMAP, Ecole Polytechnique and INRIA, being supported by a doctoral fellowship of Ecole Polytechnique.}
\begin{document}
\maketitle
\begin{abstract}
Basic properties in Perron-Frobenius theory are strict positivity, primitivity,
and irreducibility. Whereas for nonnegative matrices,
these properties are equivalent to elementary graph properties
which can be checked in polynomial time, we show that for Kraus maps
- the noncommutative generalization of stochastic matrices - 
checking strict positivity (whether the map sends the cone to its interior)
is NP-hard. The proof proceeds
by reducing to the latter problem the existence 
of a non-zero solution of a special system of bilinear equations. The complexity
of irreducibility and primitivity is also discussed in the noncommutative setting.

\end{abstract}
\section{Introduction}
Irreducibility, primitivity, and strict positivity
are basic structural notions of Perron-Frobenius theory~\cite{BermanPlemmons}.
Recall that 
a linear map $A$ leaving invariant a 
(closed, convex, and pointed) cone $C$ of a vector space
is said to be strictly {\em positive} if it sends the cone to its interior; {\em primitive} if it has a power that is positive, and {\em irreducible} if it does not leave invariant a non-trivial face of the cone. 
These notions allow one
to determine spectral or dynamical properties of the map.
In particular, the strongest of the above notions, strict positivity,
entails the strict contraction of $A$ with respect to Hilbert's projective
metric (Birkhoff's theorem), and so, the convergence of the rescaled iterates
of $A$ 
 to a rank one linear map with a geometric rate.
The latter property is of importance
in a number of applications, including ``consensus theory''
for distributed systems or population dynamics. It is natural to ask how properties of this nature can be checked for various classes of cones.

If $C$ is the
standard positive cone of $\mathbb{R}^n$,  $A$ can be identified
to a nonnegative matrix $A\in M_n(\R)$. Then, strict positivity, primitivity,
and irreducibility, can be easily checked. 
Indeed, a nonnegative
matrix $A$ is strictly positive if and only if all its entries are positive.
Moreover, $A$ is
primitive if and only if $A^{n^2-2n+2}$ is strictly positive~\cite{HornMatrix}. 
Finally, it is irreducible if and only if the associated directed graph is strongly connected. 
Note also that an efficient combinatorial algorithm is available to compute the period of an irreducible matrix, which allows one in particular to decide if it is primitive~\cite{denardo}. Therefore, primitivity and irreducibility
for nonnegative matrices 
are equivalent to well known problems of graph theory,
that can be solved in polynomial time.


Another important class of maps arises when considering the
cone $C$ of positive semidefinite matrices. 
Then, the noncommutative analogue of a stochastic matrix is a Kraus map, i.e., a completely positive and trace-preserving map on this cone.
 Kraus maps are fundamental objects in quantum control and information theory, as they represent quantum channels.
The notions of irreducibility, strict positivity and primitivity 
are of importance for Kraus maps, see in particular~\cite{Farenick96,Sanz2012,sepulchre,ReebWolf2011}. It is natural to ask whether we can verify these properties for Kraus maps in polynomial time, as in the case of nonnegative matrices.

Our main result, Theorem~\ref{theo-decicompo}, asserts that checking the strict positivity of a Kraus map is NP-hard. It may come as a surprise that strict positivity,
which is the simplest property in the case of nonnegative matrices,
turns out to be the hardest one in the case of Kraus maps.
Indeed, we derive 
from previous results that the irreducibility and primitivity of a Kraus
map can be checked in polynomial time.
A classical lemma of Burnside on matrix algebras combined
with a result of Farenick~\cite{Farenick96} implies that the irreducibility
of a completely positive map can be checked in polynomial time. Moreover,
a characterization given by Sanz, P\'erez-Garc\'ia, Wolf and Cirac~\cite{Sanz2012} also implies
that the primitivity of a Kraus map can be checked in polynomial time. See Corollary~\ref{pr-poly} below for the derivation of these two facts. 
Note that in each of these results, we assume that the input -which determines the Kraus map- consists of the Kraus operators. 


To show Theorem~\ref{theo-decicompo}, we first show that  the strict positivity of a Kraus map is equivalent to the non feasibility of the bilinear system given by the Kraus operators, or  equivalently the non-existence of a rank one matrix in the orthogonal complement of the subspace generated by the Kraus operators, see Lemma~\ref{l-stp}. Then, we prove that every 3SAT problem can be reduced in polynomial time  to the problem of checking the feasibility of a bilinear system given by a set of Kraus operators, see Theorem~\ref{th-3SATtorandkone}.

We note that several rank minimization problems have been extensively studied in the literature~\cite{Fazel04rankminimization,RechtXuHassibi,Faugere13}. In particular, the problem
of finding a matrix of minimal rank in a affine subspace is known to be NP-hard~\cite{Buss99,Recht,Donoho} and hard to approximate~\cite{Natarajan}. However, here the matrix subspace is linear instead of affine,
and rank minimization in a linear subspace is a trivial
subproblem. 
Note also that Hillar and Lim~\cite{HilarLim} showed the NP-hardness of the bilinear feasibility problem, by reducing the graph 3-Colorability problem to it. However, the bilinear systems arising from a Kraus map are special due to the
unital constraint or trace-preserving property of the Kraus map. Hence Theorem~\ref{th-3SATtorandkone} is a different result; it does not seem easy to deduce it from the NP-hardness of checking the feasibility of bilinear systems, see Remark~\ref{rem-HillarLim}.

\section{Irreducibility, primitivity and strict positivity for completely positive maps}\label{sec-irr}
Throughout the paper, the space of Hermitian matrices is denoted by $\sym_n$. Denote by $\preceq$ ($\prec$) the (strict) Loewner order on the space $\sym_n$, i.e., $A\preceq B$ ($A\prec B$) if and only if $B-A$ is a positive semidefinite (definite) matrix. The adjoint matrix (conjugate transpose) of a matrix $A\in\cC^{n\times n}$ is denoted by $A^*$. 

To a family of $n\times n$ complex matrices $V_1,\dots,V_m$, we
associate the {\em completely positive map}
$\Psi:\sym_n\rightarrow \sym_n$,
\begin{align}\label{a-PhiX}
\Psi(X):=\sum_{i=1}^m V_iXV_i^*,\quad  X\in \sym_n\enspace.
\end{align}
This map is said to be a {\em Kraus map}  if
\begin{align}\label{a-ViViIn}
\sum_{i=1}^m V_i^ *V_i=I_n\enspace ,
\end{align}
then, the matrices $V_1,\dots,V_m$ are called \firstdef{Kraus operators}.


We denote by $\cS_k(V_1,\dots,V_m)$ the complex linear space spanned by all the products of $k$ Kraus operators $\{V_1,\dots,V_m\}$:
$$
\cS_k(V_1,\dots,V_m):=\Span\{V_{i_k}\dots V_{i_1}:i_k,\dots,i_1\in\{1,\dots,m\}\}\enspace.
$$
We also denote by $\D_k(V_1,\dots,V_m)$ the complex linear space spanned by all the products of at most $k$ Kraus operators:
$$
\D_k(V_1,\dots,V_m):=\Span\{V_{i_j}\dots V_{i_1}:1\leq j\leq k, i_j,\dots,i_1\in\{1,\dots,m\}\}\enspace.
$$
We denote by $\mathcal{A}(V_1,\dots,V_m)= \cup_{k\geq 1} \D_k(V_1,\dots,V_m)$ 
the algebra generated by the Kraus operators $\{V_1,\dots,V_m\}$:
$$
\mathcal{A}(V_1,\dots,V_m):=\Span\{V_{i_k}\dots V_{i_1}: k\in \bN, i_k,\dots,i_1\in\{1,\dots,m\}\}\enspace.
$$
\begin{lemma}\label{l-pAPhiD}There is $p\leq n^2$ such that $\mathcal{A}(V_1,\dots,V_m)=\D_{p}(V_1,\dots,V_m)$.
\end{lemma}
\begin{proof}
It is clear that for all $k=1,2,\dots,$ we have
$$
\D_{k+1}(V_1,\dots,V_m)\supset \D_k(V_1,\dots,V_m) \cup \{V_iX: X\in \D_k(V_1,\dots,V_m),i\in\{1,\dots,m\}\}\enspace.
$$ Hence there is $p\leq n^2$ such that $\D_{p+1}(V_1,\dots,V_m)=\D_p(V_1,\dots,V_m)$ and thus $\cA(V_1,\dots,V_m)=\D_p(V_1,\dots,V_m)$.
\end{proof}
We next recall the definitions of irreducibility, strict positivity and primitivity for completely positive maps. 

\begin{defi}[Irreducibility~\cite{Farenick96}]
 The map $\Psi$ is irreducible if there is no face of $\sym_n^ +$ invariant by $\Psi$, where a face $\mathcal{F}$ of $\sym_n^ +$ is a (closed, convex) cone strictly contained in $\sym_n^ +$ such  that if $P\in \mathcal{F}$ then $Q\in \mathcal{F}$ for all $Q\preceq P$.
\end{defi}

\begin{defi}[Strict positivity]
 The map $\Psi$ is  strictly positive if for all $X\succeq 0$, $\Psi(X)\succ 0$.
\end{defi}
A standard compactness argument shows that $\Psi$ is strictly positive if and only if
\[
\Psi(X)\succeq \alpha \operatorname{tr}(X) I,\qquad \forall X\in \sym_n^+
\]
for some constant $\alpha>0$.

\begin{defi}[Primitivity~\cite{Sanz2012}]
 The  map $\Psi$ is primitive if there is an integer $p>0$ such that $\Psi^p$ is strictly positive.
\end{defi} 
It will be convenient to consider the following three problems.
\begin{prob}[Irreducibility of Completely Positive Maps]
\label{pb-irr}

Input: {\rm integers $n,m$, and matrices $V_1,\dots,V_m\subset \cC^{n\times n}$ with
rational entries}.

Question: {\rm Is the map $\Psi$ defined by~\eqref{a-PhiX} irreducible?}
\end{prob}
\begin{prob}[Primitivity of Completely Positive Maps]
\label{pb-prim}

Input: {\rm integers $n,m$, and matrices $V_1,\dots,V_m\subset \cC^{n\times n}$ with
rational entries}

Question: 
{\rm Is the map $\Psi$ defined by~\eqref{a-PhiX} primitive?}
\end{prob}
\begin{prob}[Strict positivity of Kraus maps]
\label{pb-pos}

Input: {\rm integers $n,m$, and matrices $V_1,\dots,V_m\subset \cC^{n\times n}$ with
rational entries, satisfying} 
\begin{align}\label{e-unital}
\sum_{i=1}^ m V_i^ * V_i=I_n \enspace .
\end{align}
Question: {\rm Is the Kraus map associated to $\{V_1,\dots,V_m\}$ strictly positive?}
\end{prob}
We next show that the first two problems can be solved in polynomial time
whereas the last one is NP-hard.

%



\section{Checking the irreducibility and primitivity is polynomial}
We shall need the following characterization of irreducibility.
\begin{prop}\label{th-irr}
The completely positive map $\Psi$ given by~\eqref{a-PhiX}
is irreducible if and only if  $\mathcal{A}(V_1,\dots,V_m)=\cC^{n\times n}$.
\end{prop}
\begin{proof}
Farenick showed in~\cite[Theorem~2]{Farenick96} that the reducibility
of $\Psi$ is equivalent to the existence of a non-trivial (other than $\{0\}$ or $\mathbb C^n$) common invariant subspace of all $\{V_i\}$. 
By Burnside's theorem on matrix algebra (see~\cite{MR2073890}), the latter property holds if and only if the algebra $\mathcal{A}(V_1,\dots,V_m)$ is 
not the whole matrix space.
\end{proof}

We shall need the following characterization of primitivity
of completely positive maps,
which is a consequence of a ``quantum version
of Wielandt inequality'' established by Sanz, P\'erez-Garc\'ia, Wolf and Cirac
for Kraus maps.
\begin{theo}[Corollary of~\cite{Sanz2012}]\label{theo-wolf}
Assume that the completely positive map $\Psi$ is irreducible.
Then, $\Psi$ is primitive if and only if there is $q\leq (n^2-m+1)n^2$ such that the space $\cS_q(V_1,\dots,V_m)$ coincides with $\cC^{n\times n}$,
for some $q\leq (n^2-m+1)n^2$. 
\end{theo}
\begin{proof}
Theorem~1 of~\cite{Sanz2012} shows that if $\Psi$ is a Kraus map, then, it is primitive if and only if $\cS_q(V_1,\dots,V_m)$ 
coincides with $\cC^{n\times n}$,
for some $q\leq (n^2-m+1)n^2$. We next show that this implies that the same
property holds for all irreducible completely positive maps.
Indeed, it follows from the Perron-Frobenius
theorem that the adjoint 
map $\Psi^*$ has an eigenvector $A$ in the cone of positive 
semidefinite matrices such that the associated eigenvalue is the spectral radius
of $\Psi$, $\rho(\Psi)$,
i.e.
\begin{align}
\sum_{1\leq i\leq m} V_i^* AV_i = \rho(\Psi)A \enspace .\label{e-kraus}
\end{align}
Since $\Psi$ is irreducible, $\Psi^*$
is also irreducible (this follows from~\cite[Theorem~2]{Farenick96}), and so this eigenvector must belong to the interior of the cone,
meaning that $A$ is a positive definite matrix. 
Now, for all invertible matrices $U$, define $\Gamma_U(X):=UXU^*$. 
Then, the map $\Phi = \rho(\Psi)^{-1} \Gamma_{A^{1/2}}\circ \Psi \circ \Gamma_{A^{-1/2}}$ satisfies 
\[
\Phi(X) = \sum_{i=1}^m W_i X W_i^*,\qquad \text{with } W_i = \rho(\Psi)^{-1/2}
A^{1/2}U_i A^{-1/2} \enspace ,
\]
and it follows from~\eqref{e-kraus} that it is a Kraus map.
Moreover, since $\cS_q(V_1,\dots,V_m)=A^{-1/2}\cS_q(W_1,\dots,W_m)A^{1/2}$,
$\cS_q(V_1,\dots,V_m)$ coincides with $\CC^{n\times n}$
if and only if $\cS_q(W_1,\dots,W_m)$ does. 
\end{proof}



\begin{coro}\label{pr-poly}
The irreducibility and the primitivity of  a completely positive map
can be checked in polynomial time. 
\end{coro}
\begin{proof}
By Proposition~\ref{th-irr} and Lemma~\ref{l-pAPhiD}, to decide if the Kraus map $\Psi$ is irreducible, we shall
compute the increasing sequence of matrix subspaces $\D_s(V_1,\dots,V_m)$, $s=1,2,\dots$,
and look for the first integer $k\leq n^ 2$ such that $\D_{k}(V_1,\dots,V_m)=\D_{k+1}(V_1,\dots,V_m)$.
For a given $s$, we shall represent $\D_s(V_1,\dots,V_m)$ by
a basis, i.e.,
 $$\D_s(V_1,\dots,V_m)=\Span\{M_1,\cdots,M_l\}$$
 where $\{M_1,\dots,M_l\}\in \cC^{n\times n}$ are linearly independent matrices.  Recall that extracting a basis from a family of rational
vectors can be done in polynomial time in the bit model.
Since $\D_{s+1}(V_1,\dots,V_m) = \Span\{V_iM_s, M_s,\; 1\leq i\leq m, \; 1\leq s\leq l\}$, it follows that we can compute inductively a basis
$M_1,\cdots,M_l$ of $\D_{s}(V_1,\dots,V_m)$, with $l\leq n^2$,
and that the number of bits needed to code the basis
elements $M_1,\dots,M_l$ remain polynomially bounded in the length of
the input. Hence, a basis representation
of the algebra  $\mathcal{A}(V_1,\dots,V_m)$ 
can be obtained in polynomial time.

Arguing
as above, a basis representation of $\cS_q(V_1,\dots,V_m)$ for some $q\leq (n^2-m+1)n^2$ can be computed
in polynomial time. Thus, to check the primitivity, we first check the irreducibility (which is a necessary condition),  and if it is satisfied,
we check the condition of Theorem~\ref{theo-wolf}.
\end{proof}

\section{Checking the strict positivity is NP-hard}\label{subsec-NP}
In this section, we study the complexity of Problem~\ref{pb-pos}: deciding if a Kraus map is strictly positive.
First we show that the strict positivity of a Kraus map is equivalent to the non-existence of rank one matrix in the orthogonal complement of the subspace spanned by the 
Kraus operators.

\begin{lemma}\label{l-stp}
The Kraus map $\Psi$
is strictly positive if and only if we cannot find two nonzero vectors $x,y \in \cC^n$ such that
\begin{align}\label{e-bilinear}
x^*V_iy=0,\enspace\forall i=1,\dots,m.
\end{align}
\end{lemma}
\begin{proof}
By definition, the map $\Psi$ is  strictly positive if and only if for all nonzero vectors $y\in \cC^n$, the matrix
$$
\Psi(yy^*)=\sum_{i=1}^m V_iyy^*V_i^*
$$
is positive definite. 
This holds if and only if for all nonzero vectors $x\in \cC^n$, 
$$
\sum_{i=1}^m x^*V_iyy^*V_i^* x=\sum_{i=1}^n |x^*V_i y|^2 >0.
$$
Therefore $\Phi$ is not strictly positive if and only if we can find nonzero vectors $x,y\in \cC^n$ such that~\eqref{e-bilinear} holds.
\end{proof}
Hence,  the strict positivity of a Kraus map (Problem~\ref{pb-pos}) is equivalent to  the non feasibility of the following bilinear system associated to the Kraus operators $\{V_1,\dots,V_m\}$.
\begin{prob}[Unital bilinear feasibility]
\label{pb-feasi}
Input: {\rm integers $n,m$, and matrices $V_1,\dots,V_m\subset \cC^{n\times n}$ with
rational entries, satisfying~\eqref{e-unital}.}
Question: {\rm is there a nonzero solution to the following bilinear system:}
$$
x^TV_iy=0,\enspace\forall i=1,\dots,m\enspace ?
$$
\end{prob}
Problem~\ref{pb-feasi} is trivially equivalent to the following problem on the existence of a rank one matrix in the orthogonal complement of the subspace   generated by the Kraus operators $\{V_1,\dots,V_m\}$.
\begin{prob}[Existence of rank one matrix]
\label{pb-rankone}
Input: {\rm integers $n,m$, and matrices $V_1,\dots,V_m\subset \cC^{n\times n}$ with
rational entries, satisfying~\eqref{e-unital}.}
Question: {\rm is there a rank one matrix in the orthogonal complement
of the subspace  spanned by $\{V_1,\dots,V_m\}$?}
\end{prob}
Consider also the following similar problem without the unital constraint on matrices:
\begin{prob}[Bilinear feasibility]
\label{pb-feasi-g}
Input: {\rm integers $n,m$, and matrices $W_1,\dots,W_m\subset \cC^{n\times n}$ with
rational entries. }
Question: {\rm is there a nonzero solution to the following bilinear system:}
$$
x^TW_iy=0,\enspace\forall i=1,\dots,m\enspace ?
$$
\end{prob}

\begin{theo}\label{th-3SATtorandkone}
 The 3SAT problem is reducible in polynomial time to Problem~\ref{pb-feasi}.
\end{theo}
The proof is based on the following observation. 
An instance of the 3SAT problem with $N$ Boolean variables $X_1,\ldots,X_{N}$ and $M$ clauses can be coded by a system of polynomial equations in $N$ complex variables $x_1,\dots, x_{N}$,
\begin{align}\label{a-e12}
\left\{\begin{array}{l}
 (1+p_{i}x_{k^1_i})(1+q_{i}x_{k^2_i})(1+r_{i}x_{k^3_i})=0,\quad i=1,\cdots,M\\
x_i^2=1,\quad i=1,\cdots,N
\end{array}\right.
\end{align}
where $k^1_i, k^2_i, k^3_i \in \{1,\dots, N\}$, $p_{i},q_{i}, r_i\in\{\pm 1\}$ and $k^1_i\neq k^2_i $ for all $1\leq i\leq M$.
 The Boolean variable $X_i$ is true if $x_i=1$ and false if $x_i=-1$.
For instance, the clause $X_1\vee \neg X_2 \vee X_4$
corresponds to the polynomial $(1-x_1)(1+x_2)(1-x_4)$ and the clause $\neg X_6\vee \neg X_1\vee  X_2$
corresponds to the polynomial $(1+x_6)(1+x_1)(1-x_2)$. 

Therefore, to prove Theorem~\ref{th-3SATtorandkone}, it is sufficient to construct in polynomial time a set of Kraus operators $\{V_1,\dots,V_m\} \subset \cC^ n$ with rational entries satisfying~\eqref{e-unital},
such that there is a solution to~\eqref{a-e12} if and only if there are two nonzero vectors $x,y\in\cC^ n$ such that`\eqref{e-bilinear} holds.

We begin by the following basic lemma.
\begin{lemma}\label{r-prop}
Let $a_k(\cdot , \cdot):\cC^n\times \cC^n \rightarrow \cC, \enspace 1\leq  k\leq M$ be a finite set of bilinear forms. 
 There is a solution $x\in \cC^n$ to the system 
$$
a_k(x,x)=0 ,\qquad 1\leq k\leq M
$$
if and only if there is a pair of non-zero vectors $x=(x_i)_{1\leq i\leq n},y=(y_i)_{1 \leq i\leq n}\in \cC^n$ satisfying the system
\begin{align}\label{a-e2}
 \left\{\begin{array}{l}
 a_k(x,y)=0,\qquad 1\leq k\leq M\\
 x_i y_j-x_jy_i=0,
\qquad  1\leq i<j\leq n 
\enspace .
\end{array}\right.
\end{align}
\end{lemma}
\begin{proof}
The last equations require that $y$ be proportional to $x$.
\end{proof}
The next lemma shows that  system~\eqref{a-e12} can be transformed into a set of homogeneous equations.
\begin{lemma}\label{l-transf}
 Let $N,M\in \mathbb N$. Let $(k_i^1)_i, (k_i^2)_i, (k_i^3)_i$ 
be three sequences of integers in $\{ 1,\cdots,N\}$. Let $(p_i)_i, (q_i)_i, (r_i)_i$ be three sequences of real numbers. 
Consider the following system of equations on the variables $(x_i)_{1\leq i\leq N}$:
\begin{align}\label{a-e6}
\left\{\begin{array}{l}
 (1+p_{i}x_{k^1_i})(1+q_{i}x_{k^2_i})(1+r_{i}x_{k^3_i})=0,\quad i=1,\cdots,M\\
x_i^2=1,\quad i=1,\cdots,N
\end{array}\right.
\end{align}
The system~\eqref{a-e6} has a solution $x\in \cC^N$ if and only if there is a pair of nonzero vectors $x=(x_i)_{0\leq i \leq N+2M},y=(y_i)_{0\leq i\leq N+2M}\in \cC^{N+2M+1}$ satisfying
the following system:
\begin{align}\label{a-e10}
\left\{\begin{array}{l}
 (x_0+p_{i}x_{k^1_i}+q_{i}x_{k^2_i}+p_{i}q_{i}x_{N+i})y_{N+M+i}=0,\quad i=1,\cdots,M\\
x_{k_i^1}y_{k_i^2}-x_0y_{N+i}=0,\quad i=1,\cdots,M\\
(x_0+r_{i}x_{k^3_i}-x_{N+M+i})y_j=0,\quad i=1,\cdots,M, \quad j=0,\dots,N+2M\\
x_iy_i-x_0y_0=0,\quad i=1,\cdots,N+M\\
 x_i y_j-x_jy_i=0,
\qquad  0\leq i<j\leq N+2M 
\end{array}\right.
\end{align}
\end{lemma}
\begin{proof}
 A simple rewriting of the system~\eqref{a-e6} is:
\begin{align}\label{a-e7}
\left\{\begin{array}{l}
 (1+p_{i}x_{k^1_i}+q_{i}x_{k^2_i}+p_{i}q_{i}x_{k^1_i}x_{k^2_i})(1+r_{i}x_{k^3_i})=0,\quad i=1,\cdots,M\\
x_i^2=1,\quad i=1,\cdots,N
\end{array}\right.
\end{align}
By introducing $2M$ extra variables, denoted by $\{x_{N+i}\}_{1\leq i\leq 2M}$,
 to replace the variables $\{x_{k_i^1}x_{k_i^2}, 1+r_{i}x_{k^3_i}\}_{i\leq M}$, we rewrite the system~\eqref{a-e7} as:
\begin{align}\label{a-e8}
\left\{\begin{array}{l}
 (1+p_{i}x_{k^1_i}+q_{i}x_{k^2_i}+p_{i}q_{i}x_{N+i})x_{N+M+i}=0,\quad i=1,\cdots,M\\
x_{k_i^1}x_{k_i^2}-x_{N+i}=0,\quad i=1,\cdots,M\\
1+r_{i}x_{k^3_i}-x_{N+M+i}=0,\quad i=1,\cdots,M\\
x_i^2=1,\quad i=1,\cdots,N+M\\
\end{array}\right.
\end{align}
We next add an extra variable $x_0$ to replace the affine term $1$ to construct a system of homogeneous polynomial equations
of degree 2:
\begin{align}\label{a-e9}
\left\{\begin{array}{l}
 (x_0+p_{i}x_{k^1_i}+q_{i}x_{k^2_i}+p_{i}q_{i}x_{N+i})x_{N+M+i}=0,\quad i=1,\cdots,M\\
x_{k_i^1}x_{k_i^2}-x_0x_{N+i}=0,\quad i=1,\cdots,M\\
(x_0+r_{i}x_{k^3_i}-x_{N+M+i})x_j=0,\quad i=1,\cdots,M, \quad j=0,\dots,N+2M\\
x_i^2-x_0^2=0,\quad i=1,\cdots,N+M\\
\end{array}\right.
\end{align}
Then that there is a solution to~\eqref{a-e8} if and only if there is a solution $x=(x_i)_{0\leq i\leq N+2M}$ to~\eqref{a-e9} such that $x_0\neq 0$.
By Lemma~\ref{r-prop}, we know that the system~\eqref{a-e9} has a solution $x=(x_i)_{0\leq i\leq N+2M}$
with $x_0\neq 0$ if and only if there 
is a pair of non-null vectors $x=(x_i)_{0\leq i\leq N+2M}$ and $y=(y_i)_{0\leq i\leq N+2M}$ with $x_0y_0\neq 0$ satisfying~\eqref{a-e10}.

So far, we proved that there is a solution to~\eqref{a-e6} if and only if there is a pair of nonzero vectors $x,y \in \cC^{N+2M+1}$ satisfying~\eqref{a-e10} such that $x_0y_0\neq 0$.
We next prove by contradiction that all nonzero pairs of solutions to~\eqref{a-e10} satisfy $x_0y_0\neq 0$.

Let $x=(x_i)_{0\leq i\leq N+2M}$ and $y=(y_i)_{0\leq i\leq N+2M}$ be a pair of nonzero solutions to~\eqref{a-e10} such that $x_0y_0=0$. Since by the last constraint in~\eqref{a-e10}, $x$ and $y$ are proportional to each other, we know that $x_0=y_0=0$. 
Suppose that there is $1\leq i_0\leq N+M$ such that $x_{i_0}\neq 0$, then by the fourth equation of~\eqref{a-e10} we know that:
$$
x_{i_0}y_{i_0}=0,
$$
thus $y_{i_0}=0$. This implies that $y$ is a zero vector because $x$ and $y$ are proportional to each other. Hence $x_i=0$ for all $i\leq N+M$. Now we apply this condition to the third equation in~\eqref{a-e10} to obtain:
$$
x_{N+M+i}y_j=0, \quad i=1,\dots,M,\enspace j=0,\dots,N+2M\enspace.
$$
If $x$ is a nonzero vector, necessarily there is $i_0$ such that $x_{N+M+i_0}\neq 0$, in that case $y$ is a zero vector. Therefore we deduce that for all nonzero solution of~\eqref{a-e10}, it is necessary that $x_0y_0\neq 0$.

\end{proof}

\begin{lemma}\label{l-rankone}
Consider the system~\eqref{a-e6} in Lemma~\ref{l-transf}. We suppose in addition that $k_i^1\neq k_i^2$ for all $1\leq i\leq M$ and
that $(p_i)_i, (q_i)_i, (r_i)_i$ are sequences of numbers in $\{\pm 1\}$.
Let $n=N+2M+1$. There is a finite family of matrices $\{V_i\}_{1\leq i\leq m}\subset \cC^{n\times n}$ with entries in $\{0,\pm 1,\pm\frac{1}{3}\}$ such that
the system~\eqref{a-e6} has a solution if and only if there is nonzero solution to the following bilinear system:
$$
x^T V_i y=0,\enspace i=1,\dots,m\enspace.
$$
Besides,
the integer $m$ can be bounded by a polynomial in $N$ and $M$ and the matrices $\{V_i\}_{1\leq i\leq m}$ satisfy:
$$
\sum_{i=1}^m V_i^*V_i=(2N+7M+4)^2 I_n
$$ 
\end{lemma}
\begin{proof}
We denote by $\{e_i\}_{0\leq i\leq N+2M}$ the standard basis vectors in $\cC^{N+2M+1}$.
We know from Lemma~\ref{l-transf}  that the system~\eqref{a-e6} admits a solution if and only if there is a pair of non-null vectors $x,y\in \cC^n$ satisfying
\begin{align}\label{a-e11}
\left\{\begin{array}{l}
 x^\top (e_0+p_{i}e_{k^1_i}+q_{i}e_{k^2_i}+p_{i}q_{i}e_{N+i})e_{N+M+i}^\top y=0,\quad i=1,\cdots,M\\
x^\top (e_{k_i^1}e_{k_i^2}^\top -e_0e_{N+i}^\top )y=0,\quad i=1,\cdots,M\\
x^\top (e_0+r_{i}e_{k^3_i}-e_{N+M+i})e_j^\top y=0,\quad i=1,\cdots,M, \quad j=0,\dots,N+2M\\
x^\top (e_ie_i^\top -e_0e_0^\top )y=0,\quad i=1,\cdots,N+M\\
 x^\top (e_i e_j^\top -e_je_i^\top )y=0,
\qquad  0\leq i<j\leq N+2M 
\end{array}\right.
\end{align}
The system~\eqref{a-e11} has $N+3M+(N+2M+1)(4M+N)/2$ bilinear equations. Let $m_0=N+3M+(N+2M+1)(4M+N)/2$ and denote by $\{A_i\}_{1\leq i\leq m_0}$ the matrices corresponding to the $m_0$ bilinear forms in~\eqref{a-e11}.  Recall that $(p_i)_i, (q_i)_i,(r_i)_i$ are sequences of numbers in $\{1,-1\}$.
Therefore we transformed
the system~\eqref{a-e6} to the following bilinear system:
\begin{align}\label{a-xAiy}
x^T A_i y=0,\enspace i=1,\dots,m_0\enspace,
\end{align}
where $A_i$ have entries in $\{0,1,-1\}$. We check the five lines in~\eqref{a-e11} and obtain that
$$
\begin{array}{ll}
\displaystyle\sum_{i=1}^{m_0} A_i ^*A_i = &\displaystyle\sum_{i=1}^M 4 e_{N+M+i}e_{N+M+i}^\top +\displaystyle\sum_{i=1}^M (e_{k^2_i}e_{k^2_i}^\top +e_{N+i}e_{N+i}^\top )\\
&+\displaystyle\sum_{i=1}^M\sum_{j=0}^{N+2M} 3e_je_j^\top +\displaystyle\sum_{i=1}^{N+M} (e_ie_i^\top +e_0e_0^\top )\\
&+\displaystyle\sum_{i<j} (e_je_j^\top +e_ie_i^\top )
\end{array}
$$
Therefore we have that
$$
\sum_{i=1}^{m_0}A_i^*A_i=\left(\begin{array}{llll}k_1 & & & \\  & k_2 & & \\  &  &\ddots &  \\ & & & k_n \end{array}\right)
$$
where $k_i \leq 2N+7M+4$ for all $1\leq i\leq n$. Remark that due to the third line of equations in~\eqref{a-e11}, for each $0\leq j\leq N+2M$, there is an integer $1\leq n_j\leq m_0$  such that
$$
A_{n_j}^*A_{n_j}=3e_je_j^\top .
$$
By letting $B_j={A_{n_j}}/{3}$ we get that:
$$
3 B_j^*B_j=e_je_j^\top.
$$
For all $1\leq j\leq n$ let $l_j=(2N+7M+4)^2-n_j$. 
Let $m=m_0+3\sum_{j=1}^n l_j$ and $\{V_i\}_{1\leq i\leq m}$ be the sequence of matrices containing $\{A_i\}_{1\leq i\leq m_0}$ 
and $3l_j$ times the matrix $B_j$ for all $1\leq j\leq n$.  
Then we have
$$
\sum_{i=1}^mV_i^*V_i=\sum_{i=1}^{m_0} A_i^*A_i+\sum_{j=1}^{n}3l_jB_j^*B_j=(2N+7M+4)^2 I_n.
$$
Since for all $1\leq j\leq n$, $B_j$ is co-linear to a matrix in $\{A_i\}_{i\leq m_0}$. The feasibility of the system
\begin{align}\label{a-xViy}
x^T V_i y=0,\enspace i=1,\dots,m
\end{align}
is equal to that of~\eqref{a-xAiy}. Thus the system~\eqref{a-e6} admits a solution if and only if there is a nonzero solution to~\eqref{a-xViy}.
\end{proof} 

We now prove Theorem~\ref{th-3SATtorandkone}.
\begin{proof}
Let $k^1_i, k^2_i, k^3_i \in \{1,\dots, N\}$, $p_{i},q_{i},r_{i}\in\{\pm 1\}$ and $k^1_i\neq k^2_i $ for all $1\leq i\leq M$ such that the system
\begin{align}\label{a-e122}
\left\{\begin{array}{l}
 (1+p_{i}x_{k^1_i})(1+q_{i}x_{k^2_i})(1+r_{i}x_{k^3_i})=0,\quad i=1,\cdots,M\\
x_i^2=1,\quad i=1,\cdots,N
\end{array}\right.
\end{align}
 corresponds to an instance of 3SAT problem with $N$ Boolean variables and $M$ clauses.
  By Lemma~\ref{l-rankone}, we can construct in polynomial time (with respect to $N$ and $M$) a sequence of $n\times n$ matrices $\{V_i\}_{1\leq i\leq m}$ with entries in $\{0,\pm \frac{1}{l}, \pm \frac{1}{3l}\}$ where $l=(2N+7M+4)$ such that there is a solution to~\eqref{a-e122} if and only if there is a nonzero solution to the bilinear system~\eqref{a-xViy}.
Besides, the matrices $\{V_i\}_{1\leq i\leq m}$ satisfy~\eqref{e-unital}.
\end{proof}
We deduce the complexity of Problem~\ref{pb-pos} from Theorem~\ref{th-3SATtorandkone} and Lemma~\ref{l-stp}.
\begin{theo}\label{theo-decicompo}
 Deciding whether a Kraus map  is strictly positive (Problem~\ref{pb-pos})
is NP-hard.
\end{theo}
\begin{rem}\label{rem-HillarLim}
Hillar and Lim~\cite{HilarLim} obtained the NP-hardness of Problem~\ref{pb-feasi-g} by reducing graph 3-Colorability problems to it. Let $\{W_1,\dots,W_m\}\subset \cC^{n\times n}$ be arbitrary matrices and consider the bilinear system:
$$
x^T W_i y=0,\enspace \forall i=1,\dots,m\enspace.
$$
Let $U\in \cC^{n\times n}$ be any matrix such that
\begin{align}\label{a-WiWiU}
\sum_{i=1}^m W_i^* W_i=U^*U\enspace.
\end{align}
If $U$ is not invertible, than the intersection of the null spaces of $\{W_1,\dots,W_m\}$ is not empty and the latter bilinear system is clearly feasible. If $U$ is invertible, than the latter bilinear system is feasible if and only if
 the following bilinear system is feasible:
$$
x^T W_i U^{-1} y=0,\enspace \forall i=1,\dots,m\enspace.
$$
Hence every instance of Problem~\ref{pb-feasi-g}  can be reduced to an instance of Problem~\ref{pb-feasi} by computing the matrix $U\in \cC^{n\times n}$ satisfying~\eqref{a-WiWiU}.
However, in general such a matrix $U$ does not have rational entries. Therefore, it is not obvious to deduce the complexity of Problem~\ref{pb-feasi} in the bit model from the NP-hardness of bilinear feasibility.
In this respect, the proof of Theorem~\ref{th-3SATtorandkone} should be compared with the one of Hillar and Lim~\cite{HilarLim} proving the latter result.
In order to reduce a 3-Colorability problem to a bilinear system, they use 
cubic roots of the unity to encode the three colors. Some auxiliary variables are also introduced in order to obtain a homogeneous system. However, their construction does not allow to obtain in polynomial time matrices satisfying the constraint~\eqref{e-unital}.
\end{rem}

\if and only ifalse
We also obtain the following complexity results from the above reasoning.
\begin{theo}\label{th-biliNPH}
Let the input be a finite set of matrices $\{V_1,\dots V_m\} \in \cC^{n\times n}$. The following problems are $NP$-hard:
\begin{itemize}
\item Checking if there is a nonzero solution $x\in \cC^n$ to the system
$$x^\top V_i x=0 \enspace,\enspace i=1,\dots,m\enspace.$$
\item Checking if there is a nonzero solution $x,y\in \cC^n$ to the bilinear system
$$x^\top V_i y=0 \enspace,\enspace i=1,\dots,m\enspace.$$
\item Checking if there is a rank one matrix in the orthogonal complement of the space generated by $\{V_1,\dots,V_m\}$.
\end{itemize}
\end{theo}

\fi


 \bibliographystyle{alpha}
\bibliography{../../bibliographe/biblio}
\end{document}